\documentclass{article}
\usepackage[utf8]{inputenc}
\usepackage{authblk}
\pdfoutput=1
\usepackage{multicol}
\usepackage[colorlinks]{hyperref}
\usepackage{amsmath}
\usepackage{amsthm}
\usepackage{enumerate}
\usepackage{tikz}
\usepackage{graphicx}
\usepackage{subfig}
\usepackage[width=1.\textwidth]{caption}

\usetikzlibrary{shapes.geometric}
\newtheorem{theorem}{Theorem}

\newcommand{\pSize}{3.5cm} 
\newenvironment{proof*}[1]{\noindent{\bf #1}}{\qed}

\title{How Quantum Information can improve Social Welfare}

\author[1]{Berry Groisman}
\author[2]{Michael Mc Gettrick}
\author[3]{Mehdi Mhalla}
\author[4]{Marcin Pawlowski}
\affil[1]{Department of Applied Mathematics and Theoretical Physics, University of Cambridge, UK}
\affil[2]{School of Mathematics, Statistics and Applied Mathematics, National University of Ireland, Galway, Ireland}
\affil[3]{Universit\'e Grenoble Alpes, Grenoble, France}
\affil[4]{International Centre for Theory of Quantum Technologies,
University of Gdansk, Wita Stwosza 63, 80-308 Gdansk, Poland}

\usepackage{graphicx}

\newcommand{\ket}[1]{|#1\rangle}

\begin{document}

\maketitle
\abstract{In \cite{belief1,pappa,belief3} it has been shown that quantum resources can allow us to achieve 
a family of equilibria that can have sometimes a better 
social welfare, while guaranteeing privacy.
We use graph games to propose a way to build non-cooperative games
from graph states, and we show how to achieve an unlimited improvement
with quantum advice compared to classical advice.}

\section{Introduction}
An important tool in analysing games is the concept of 
{\it Nash equilibrium} \cite{nash1950equilibrium}, which represents situations
where no player has incentive 
to deviate from their strategy. This corresponds to situations observed in real life,
with applications in economics, sociology, international relations, biology, etc.
All equilibria do not have the same  {\it social
welfare}, i.e. the average payoff is different from one equilibrium to another.
Games of incomplete information can exhibit better equilibria if players use a resource – a general correlation, $Q$. Such correlation 
can be viewed as a resource produced by a mediator to
give {\it advice} to the players. The concept of advice generalizes the notion of Nash equilibrium to a broader class of equilibria \cite{aumann}. All such equilibria can be classified according to the properties of the resource correlation. Three classes can be identified in addition to Nash equilibria (no correlation),  
namely general communication equilibria (Comm) \cite{forges1982}, where $Q$ is unrestricted, belief-invariant equilibria (BI) \cite{forges1993,forges2006,lehrer2010,liu2015} and correlated equilibria (Corr) \cite{aumann}. The canonical versions of these equilibria form a sequence of nested sets within the set of canonical correlations: 
\begin{equation*}
\text{Nash}\subset\text{Corr}\subset\text{BI}\subset\text{Comm}. 
\end{equation*}
It was demonstrated that there exist games where BI equilibria can outperform $\text{Corr}$ equilibria \cite{pappa} (in terms of a social welfare (SW) of a game) as well 
as games where BI equilibria outperform any non-BI equilibria. 

Winter at al. \cite{belief1} introduce quantum correlated equilibria as a 
subclass of BI equilibria and show 
that quantum correlations can achieve optimal SW. This provides the link with quantum nonlocality, where quantum resources are used to produce {\it non-signalling} correlations. 
In this context, belief invariance describes the largest class of correlations that obey {\it relativistic causality}.

A characteristic feature of belief-invariance is that it ensures privacy -- the other players involved in the game have 
no infomation about the input one player sent to the resource.

To obtain the canonical form of the games, \cite{mathieu2018separating} show that one can suppose that the output of the correlation resource is the answer 
the players give by delegating the extra computation (from game question to input
to the box and from output of the box to players' answer) to the mediator. 
Therefore, 
quantum equilibria can be reached in a setting where players each measure quantum 
systems or, equivalently, by just having a central 
system providing advices by measuring a quantum device.

Ref. \cite{belief1} highlights several open questions. In particular, 
\begin{enumerate}[(1)]
 \item Whether any full-coordination game (a.k.a. a {\it non-local game} in
 quantum physics and computer science communities) can be converted into a 
 conflict-of-interests
 game. Ref. \cite{pappa} gives an example of a two-player variant of the CHSH game,
 while \cite{belief1} extends their result to an $n$-player game in which there
 exists
 a BI equilibrium which is better than any $\text{Corr}$ equilibrium. 
\item How can we get a large separation between the expected payoff for the
quantum and correlated 
equilibrium cases, and what is the upper bound for the separation? In the case of 
two-player full coordination games this question was settled in
\cite{buhrman2011, junge2010}.
Are there conflict-of-interest games which exhibit large separation?
 \end{enumerate}

In this paper, we provide a natural way to convert graph games 
(and more generally stabiliser games) into conflict-of-interest games, 
and we show how we can create 
unbounded separation by increasing the number of players or using penalty techniques (a negative payoff). 


An interesting feature in these games compared to the usual
pseudo-telepathy scenarios studied in quantum information is the notion of {\it involvement} \cite{MCTX, mathieu2018separating},
which allows one to define some interesting 
scenarios in non-cooperative games and which exhibits novel features,
e.g. unlimited separation. If a player participates in the game but 
is not involved (on a particular 
round) it means that their strategy is not taken into 
account when determining the win/lose outcome. 
However, they do receive a corresponding payoff. 

Using these games one can build  games with bounded personal utilities
$v_0$, $v_1$ on $O(log(\frac 1 \epsilon))$ players
 ensuring $\frac{CSW(G)}{QSW(G}\le \epsilon$, where CSW/QSW are
 the Classical/Quantum Social Welfares, respectively.





The paper is organized as follows. In Sec. \ref{sec:graph games} we describe graph games which are the underlying non-local games used to define our games. 
In Sec. \ref{sec:non-coll games} we define a non-collaborative game as a modification of the collaborative games by introducing unequal payoffs 
corresponding to answers 0 and 1 of each player, 
and discuss the corresponding quantum perfect strategy. We consider a particular version of graph games from the cycle on five vertices. Sec. \ref{sec:variant} discusses
variations of non-collaborative games based on the cycle on five vertices. Finally, Sec. \ref{sec:amp} shows how one  can amplify the quantum advantage  by  adding a penalty for wrong answers and  by increasing the number of players.


\section{Graph games}\label{sec:graph games}
Non-local games play a key role in Quantum Information theory. They can be viewed 
as a setting in which players that are not allowed to communicate receive some inputs and have to produce some outputs, and there is a winning/losing condition depending globally on their outputs for each input. 
Particular types of games are  pseudo-telepathy games \cite{pseudo}  
which are games that can be won perfectly using quantum resources  but  that are 
impossible to win perfectly without communication  when the players have
access only to shared randomness.  
Multipartite collaborative games ($MCG(G)$) are a
family of pseudo-telepathy games based on certain types of quantum states 
called {\it graph states}.
%
%
The players are identified with vertices of the graph and have
a binary input/output each with the winning/losing conditions
built using the stabilisers of the graph states.

The combinatorial game\footnote{without considering probability distributions} 
$MCG$  with $n$ players  consists in asking the players questions:  for each question $q$, 
each player $i$ receives  one bit  $q_i$ as input and answers one  bit  $a_i$. 
They can either all win or all lose depending on their answer, 
with winning/losing conditions described by a set  $\{(q,I(q),b(q))\}$ where
\begin{itemize}
\item  $q\in \{0,1\}^n$ is a valid question in which each player $i$ 
gets the bit $q_i$  and in the subgraph of the vertices corresponding to players
receiving one, all vertices  have even degree.  Let $I_1=\{i, q_i=1\}$ 
and $G'=G_{|I_1}$, a question is valid if each vertex of $ G'$ has an even 
number of neighbors in $G'$ 
\item $I(q)\subset [n]$ is a subset of players that are called {\it `involved'}
in the question as the sum (modulo $2$) of their answers determines
the winning/losing condition according to the bit  $b(q)$:
\item $b(q)$ is defined such that the players win the game when the question 
is $q$ if the sum of the answers of the involved players is equal
to the parity of the number of edges of the subgraph of the vertices
corresponding to players receiving one:  $\sum_{i\in I(q)} a_i=b(q)=|E(G')| \bmod 2$.
\end{itemize}


 For instance the game associated to the cycle on 5 elements $MCG(C_5)$
 is defined by
 
 \begin{itemize}
 \item  When the question is  $q=11111$ (each player has input 1),
 the players lose if the binary sum of their answer is 0, {\it{i.e.}} $\sum_{i=0}^{4} a_i=0 \bmod 2$
 , and win otherwise.
 \item When the question contains $010$ for three players corresponding to three adjacent vertices, the players lose if the binary sum 
 of the answer of these three players is 1 {\it{i.e.}} $a_{i-1}+a_i+a_{i+1}=0 \bmod 2$ when $q$ contains $0_{i-1}1_i0_{i+1}$.
\item The players win otherwise.
\end{itemize}

A variation of this game can be done by reducing the set of valid questions, 
for instance in the above set-up the questions of the second type have only 
three players ``involved", so a first version could be to chose only 5 questions
of the second type and give always 0 as advice to the non-involved players. 
This is the game studied as an example in \cite{mathieu2018separating}.

 An important point is that the notion of involvement in $MCG$ games is absent 
 in unique games and introduces situations  where the players might change their 
 strategy (answer) without changing the winning/losing status
 of the global strategy.


To analyse these games and the strategies, one can imagine a scenario 
where there is one special player representing Nature who is playing against
the other 
players. The strategy of Nature is therefore a probability distribution  
over the questions that we study here (as is standard in game theory)
as a known function on the set of questions $w:T\rightarrow [0,1]$ such 
that $\sum_{t\in T} w(t)=1$.
The games will be therefore defined by equipping the combinatorial 
game with a probability distribution over the questions.

\section{Defining non-collaborative games}\label{sec:non-coll games}

Like in multipartite collaborative graph games $MCG(G)$, we associate a
non-collaborative game $NC(G)$ to each graph. We differentiate the payoff
of the 
players using the value of their output: If the global answer wins in the 
non-local game, each player gets $v_1$ if they answer 1 and $v_0$ if they
answer 0. If the global answer loses, they get 0.

To match the traditional terminology used in game theory the output from now on will be called {\it strategy}, and the input called {\it type}.  The payoff is called {\it 
utility} and the social welfare is the average of the utilities over the players.

 A non-collaborative game $NC(G)$ is thus defined from $MCG(G)$  as follows
\begin{itemize}
\item  The  considered types  are $T\subset \{0,1\}^n$ where $n$ is the number of vertices of $G$.  
\item As in $MCG$, to each type $t\in T$ corresponds an associated 
involved set $I(t)$  of players, and an expected binary answer $b(t)$.
\item As in $MCG$, the losing set is
$${\cal L}=\{(s,t), \sum_{i\in I(t)} s_i \neq b(t) \bmod 2\}.$$
We say that the players  using a strategy $s$, 
given a type $t$, collectively win the game  when the sum of 
the local strategies of the involved players is equal  to the requested
binary answer modulo 2.
\item the payoff function is: 
$$u_j(s|t)=\left\{
\begin{array}{ll}
 v_{s_j} & \mathrm{ \, \, if \, \,}  (s,t)\not \in {\cal L} \\
  0 &\mathrm{\, \,  Otherwise}
  \end{array} \right.
$$
 
\end{itemize}

Firstly we consider the
 cycle on five vertices $C_5$. We define $NC_{00}(C_5)$ 
 based on the 
 non-local
game $MCG(C_5)$  studied in \cite{MCTX,mathieu2018separating}
For questions 
 which  
involve three players, both non-involved players have type $0$ 
(see Figure \ref{fig1}).

\begin{figure}[ht]
\centering
\begin{tabular}{ccc}
\includegraphics{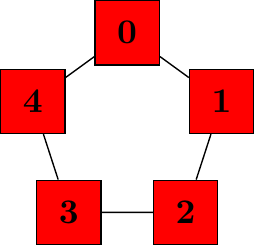}
&
\includegraphics{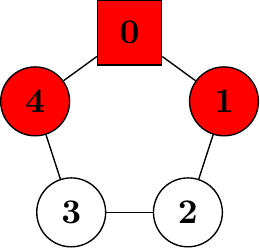}
&
\includegraphics{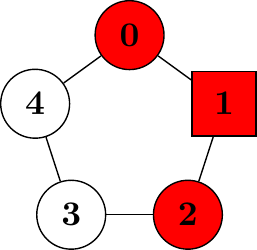}
\\
\parbox{\pSize}{\centering
(a) $T_a=11111,$\\
$ I = \{0,1,2,3,4\}, b=1$
}
&
\parbox{\pSize}{\centering
(b) $T_0=10000,$\\
$ I = \{4,0,1\}, b=0$
}
&
\parbox{\pSize}{\centering
(c) $T_1=01000,$\\
$I = \{0,1,2\}, b=0$
}
\\
\\
\includegraphics{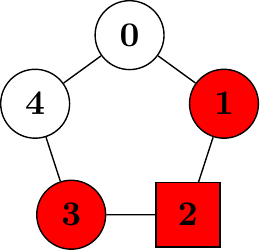}
&
\includegraphics{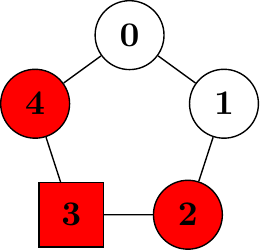}
&
\includegraphics{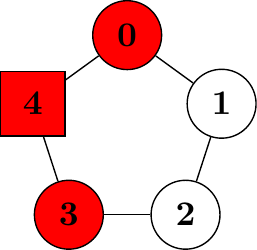}
\\
\parbox{\pSize}{\centering
(d) $T_2=00100,$\\
$I = \{1,2,3\}, b=0$
}
&
\parbox{\pSize}{\centering
(e) $T_3=00010,$\\
$I = \{2,3,4\}, b=0$
}
&
\parbox{\pSize}{\centering
(f) $T_4=00001,$\\
$I = \{3,4,0\}, b=0$
}
\\
\end{tabular}
\caption{$NC_{00}(C_5)$: Square nodes indicate
a 1 in the associated type, while circular nodes indicate a 0. \emph{Involved} players in each case are shaded in {\bf \color{red} red}.\label{fig1}}
\end{figure}

\begin{table}[!ht]
\begin{center}
\begin{tabular}{ccc} \hline\noalign{\smallskip}
Type&Involved set& Binary answer {\smallskip} \\
\hline\hline\noalign{\smallskip}
$T_a=11111$ & $I(T_a)=\{0,1 ,2,3,4\}$ & $b(T_0)=1$\\
\hline\noalign{\smallskip}
$T_0=10000$ & $I(T_0)=\{0,1 ,4\}$ & $b(T_0)=0$\\
\hline\noalign{\smallskip}
$T_1=01000$ & $I(T_1)=\{0,1 ,2\}$ & $b(T_1)=0$ \\
\hline\noalign{\smallskip}
 $T_2=00100$ & $I(T_2)=\{1 ,2,3\}$ & $b(T_2)=0$ \\
\hline\noalign{\smallskip}
 $T_3=00010$ & $I(T_3)=\{2,3 ,4\}$ & $b(T_3)=0$ \\
 \hline\noalign{\smallskip}
 $T_4=00001$ & $I(T_4)=\{3,4 ,0\}$ & $b(T_4)=0$ \\
\hline
\end{tabular}
\end{center}
\caption{$NC_{00}(G)$ game.} 
\label{T1}
\end{table}

We consider the game with the type probability distribution  $w(t)=1/6$ for all the types.

The quantum perfect strategy for $NC(G)$ is obtained when the players
each have  a qubit from graph state $\ket{G}$ \cite{MCTX}.
Each player $i$ measures their qubit according to their type $t_i$, 
getting a quantum advice representing their part of the quantum 
strategy $s_i$.\cite{MCTX}
From the study of $MCG(G)$ we have 
\begin{theorem}\label{thm1}
If all the players collaborate (follow the quantum advice) then
for any probability distribution over the types, the utility of each player
is $(v_0+v_1)/2$.
\end{theorem}
\begin{proof}
The output of each quantum measurement provides uniformly all the possible answers.
\end{proof}

\subsection{Is the quantum pseudo-telepathy solution a Nash equilibrium?}\label{0}

As the players now have an incentive to answer $1$, they can sacrifice 
 always getting a good answer to maximize their utility.
Indeed, in the previous game, each player is always involved when they
get type $1$ and with probability $1/2$ when they get type $0$; 
getting the wrong answer in that case only costs $v_0$.

Without loss of generality we consider $v_1\ge v_0$. The players now have an 
incentive to answer $1$, because they might be able to maximize their
utility
by allowing the non-zero probability of a wrong answer. Indeed, in the 
previous game, $NC_{00}(C_5)$, if the player gets type $1$ then they are 
certain that they 
are involved, and they won't gain by defecting (not following advice). However,
if their type is $0$, then the probability of them being involved is $1/2$, 
and so there is a 
fifty percent chance that they will benefit from always answering $1$ while 
not 
compromizing the winning combination. Getting the wrong answer in that case only 
costs $v_0$.

\begin{theorem}\label{thm2}
Let $p_{\text{inv}}^{(i)}(t_i,s_i)$ be the probability for the player $i$ who gets
type $t_i$ and advice $s_i$ to be involved

Then, in $NC(G)$, the quantum advice gives a belief-invariant Nash equilibrium iff
 $$\frac{v_0} {v_1}\ge (1-p), $$ 
where  \[p=\min_i\min_{t_i}p_{\text{inv}}^{(i)}(t_i,0).\] 
\end{theorem} 

\begin{proof}
If the advice is $s_i=1$ then the winning payoff is already $v_1$. Consider the case when player $i$ is given the advice $s_i=0$ (which would lead to payoff $v_0$ 
in the winning case).
If the player defects then the difference of utility is 
$-v_0  p_{\text{inv}}^{(i)}(t_i,0) + (1-p_{\text{inv}}^{(i)}(t_i,0)) (v_1-v_0)$.
So the strategy is a Nash-equilibrium when  
$ (1-p_{\text{inv}}^{(i)}(t_i,0)) v_1\le v_0   $, i.e  
$ v_0/v_1 \ge   {1-p_{\text{inv}}^{(i)}(t_i,0)} $. 
This inequality has to hold for all types and all players.

\end{proof}


 For $NC_{00}(C_5)$, 
$p_{\text{inv}}^{(i)}(0,0)=1/2$ and therefore 
the quantum nonlocal strategy is an equilibrium
only when 
$v_0/v_1 \ge1/2$.


One important characteristic of an equilibrium is the {\it Social Welfare},
which is the average utility of the players.

As a direct consequence of Theorem \ref{thm1} the average social welfare of 
the quantum strategy is independent on the graph $$QSW(NC(G))=\frac{v_0+v_1} {2}.$$

Note that the non collaborative games defined have a special feature that we call {\it guaranteed value}: in any run of the game players following the quantum strategy  receive their expected payoff with probability 1.

\section{Some versions of \texorpdfstring{$NC(C_5)$}{Lg}}\label{sec:variant}

In this section we study the game $NC_{00}(C_5)$ and then introduce a number of modifications
in order to improve the quantum advantage 
(ratio of quantum social welfare to correlated social welfare) and also to symmetrize the game such that the players get $0$ and $1$ with same probability or
have the same probability of being involved  regardless of whether their  type is 0 or 1.


\subsection{Study of \texorpdfstring{$NC_{00}(C_5)$}{Lg}}

Pure Nash equilibria  can be described by local functions: 
each player having one local type bit and one strategy bit to produce,
can locally act as follows:

\begin{itemize}
\item $ 0\rightarrow 0$ , $ 1\rightarrow 0$  constant function 0 denoted {\bf 0}
\item $ 0\rightarrow 1$ , $ 1\rightarrow 1$  constant function 1 denoted  {\bf 1}
\item $ 0\rightarrow 0$ , $ 1\rightarrow 1$  Identity  function  denoted  {\bf 2}
\item $ 0\rightarrow 1$ , $ 1\rightarrow 0$  NOT   function  denoted  {\bf 3}
\end{itemize}


The set of pure Nash  equilibria depends on  the ratio $v_0/v_1$.  
The are  20/25/40 pure Nash equilibria (4/4/6 up to symmetry) when $v_0/v_1$ lies within the interval
$[0,1/3]$, $[1/3,1/2]$ or $[1/2,1]$ respectively (see Table \ref{T:C00}). 




\begin{center}
 
 $\begin{array}{lllll|lllll|l} 
 \multicolumn{5}{c}{\text{Local functions}} 
 &  \multicolumn{5}{c}{\text{Players utility  $[\times 6]$ }}&  SW [\times30]\\ \hline\hline 
\noalign{\smallskip}
\multicolumn{11}{c}{v_0/v_1\le 1/3}{\smallskip}\\ \hline \\
{ \bf
2}
&
{ \bf
1}
&
{ \bf
1}
&
{ \bf
1}
&
{ \bf
1
}
&
2
v_0
\, + \, 
v_1
&
3
v_1
&
3
v_1
&
3
v_1
&
3
v_1
&
2v_0 +13v_1
\\

{ \bf
3
}
&
{ \bf
3}
&
{ \bf
1}
&
{ \bf
1}
&
{ \bf
1}
&
2
v_0
\, + \, 
v_1
&
2
v_0
\, + \, 
v_1
&
3
v_1
&
3
v_1
&
3
v_1
&
v_0 +11v_1
\\
{ \bf
3}
&
{ \bf
1}
&
{ \bf
3}
&
{ \bf
1}
&
{ \bf
1}
&
2
v_0
\, + \, 
v_1
&
3
v_1
&
2
v_0
\, + \, 
v_1
&
3
v_1
&
3
v_1
&
{4v_0 +11v_1} 
\\
{ \bf
3}
&
{ \bf
3}
&
{ \bf
3}
&
{ \bf
3}
&
{ \bf
1}
&
2
v_0
\, + \, 
3
v_1
&
2
v_0
\, + \, 
3
v_1
&
2
v_0
\, + \, 
3
v_1
&
2
v_0
\, + \, 
3
v_1
&
5
v_1
&
8v_0 +17v_1
\\

\hline 
\noalign{\smallskip}
\multicolumn{11}{c}{1/3 \le v_0/v_1\le 1/2}{\smallskip}\\ \hline \\
{ \bf
1}
&
{ \bf
3}
&
{ \bf
1}
&
{ \bf
1}
&
{ \bf
0}
&
5
v_1
&
2
v_0
\, + \, 
3
v_1
&
5
v_1
&
5
v_1
&
5
v_0
&
7v_0 +18v_1
\\
{ \bf
2}
&
{ \bf
2}
&
{ \bf
1}
&
{ \bf
1}
&
{ \bf
1}
&
3
v_0
\, + \, 
2
v_1
&
3
v_0
\, + \, 
2
v_1
&
5
v_1
&
5
v_1
&
5
v_1
&
6v_0 +19v_1
\\
{ \bf
3}
&
{ \bf
3}
&
{ \bf
1}
&
{ \bf
1}
&
{ \bf
1}
&
2
v_0
\, + \, 
v_1
&
2
v_0
\, + \, 
v_1
&
3
v_1
&
3
v_1
&
3
v_1
&
4v_0 +11v_1
\\
{ \bf
3}
&
{ \bf
3}
&
{ \bf
3}
&
{ \bf
3}
&
{ \bf
1}
&
2
v_0
\, + \, 
3
v_1
&
2
v_0
\, + \, 
3
v_1
&
2
v_0
\, + \, 
3
v_1
&
2
v_0
\, + \, 
3
v_1
&
5
v_1
&
8v_0 +17v_1
\\



\hline 
\noalign{\smallskip}
\multicolumn{11}{c}{v_0/v_1\ge 1/2}{\smallskip}\\ \hline \\
{ \bf
3}
&
{ \bf
2}
&
{ \bf
1}
&
{ \bf
1}
&
{ \bf
0}
&
2
v_0
\, + \, 
3
v_1
&
4
v_0
\, + \, 
v_1
&
5
v_1
&
5
v_1
&
5
v_0
&
11v_0 +14v_1
\\
{ \bf
1}
&
{ \bf
3}
&
{ \bf
1}
&
{ \bf
1}
&
{ \bf
0}
&
5
v_1
&
2
v_0
\, + \, 
3
v_1
&
5
v_1
&
5
v_1
&
5
v_0
&
7v_0 +18v_1
\\
{ \bf
2}
&
{ \bf
2}
&
{ \bf
1}
&
{ \bf
1}
&
{ \bf
1}
&
3
v_0
\, + \, 
2
v_1
&
3
v_0
\, + \, 
2
v_1
&
5
v_1
&
5
v_1
&
5
v_1
&
6v_0 +19v_1
\\
{ \bf
3}
&
{ \bf
3}
&
{ \bf
1}
&
{ \bf
2}
&
{ \bf
1}
&
2
v_0
\, + \, 
3
v_1
&
2
v_0
\, + \, 
3
v_1
&
5
v_1
&
4
v_0
\, + \, 
v_1
&
5
v_1
&
8v_0 +17v_1
\\
{ \bf
3}
&
{ \bf
3}
&
{ \bf
3}
&
{ \bf
3}
&
{ \bf
1}
&
2
v_0
\, + \, 
3
v_1
&
2
v_0
\, + \, 
3
v_1
&
2
v_0
\, + \, 
3
v_1
&
2
v_0
\, + \, 
3
v_1
&
5
v_1
&
8v_0 +17v_1
\\
{ \bf
3}
&
{ \bf
2}
&
{ \bf
3}
&
{ \bf
2}
&
{ \bf
2}
&
2
v_0
\, + \, 
3
v_1
&
4
v_0
\, + \, 
v_1
&
2
v_0
\, + \, 
3
v_1
&
3
v_0
\, + \, 
2
v_1
&
3
v_0
\, + \, 
2
v_1
&
14v_0 +11v_1
\\
\hline 
\end{array}
$
\captionof{table}{Nash equilibria for three intervals of the value $v_0/v_1$. Note that the critical values 1/2 and 1/3 have union of both tables as equilibria.}
\label{T:C00}

\end{center}

We can see that most of these equilibria  (all of them when $v_0/v_1\ge 1/2$) 
correspond to local functions winning for the 5  types. 

When $v_0=2/3$ and $v_1=1$ then  the quantum social welfare of the
pseudotelepathy strategy is $QSW=0.83$  whereas the best classical social 
welfare  $CSW=0.77$.

As noted in section \ref{0} the probability of being involved in $NC_{00}$ is $(p(1,s)=1$ and $p(0,s)=1/2$ and the quantum pseudotelepathy measurements 
strategy is an equilibrium if $ v_0/v_1\ge 1/2$.

Simililar behavior can be seen with Pareto equilibria (ones in which local utility cannot improve without reducing the outcome of someone else): see Appendix.

Recall that the characteristic feature of $NC_{00}(C_5)$ is that each player has unequal probabilities of getting different types. The game can be symmetrized 
by changing the types of the non-involved players from $00$ to $01$, as shown in the next section.
\subsection{Comments on \texorpdfstring{$NC_{01}(C_5)$}{Lg}}
We define  a second variant from $MCG(C_5)$ : $NC_{01}(C_5)$ where any player
gets the types 0 and 1 with probability 1/2 by adding an 
extra 1 for a non-involved player in the types so that 
$T_i=0_{i-1}1_i0_{i+1}1_{i+2}0_{i+3}$: see Table \ref{T:C01}.



If the type probability distribution is $w(t)=1/6$ for all the types, then
one can see that any player is involved 
with probability 2/3 whether their input is 0 or 1, i.e. 
$p_{\text{inv}}^{(i)}(0,0)=p_{\text{inv}}^{(i)}(1,0)=2/3$. 
Hence, by Theorem \ref{thm2}, the quantum strategy of MCG  produces a 
Nash equilibrium  iff $v_0/v_1 \ge 1/3$. 
Thus, one of the benefits of this variant is that quantum Nash equilibria exist at a lower ratio $v_0/v_1$.


Note that in this version each player is getting a perfect random bit as 
advice : $p(a=1)=p(a=0)=1/2$. 


When $v_0=2/3$ and $v_1=1$ then  the quantum social welfare of the
pseudotelepathy strategy is $QSW=0.83$  whereas 
the best classical social welfare  is $CSW=0.78$.

 

\begin{table}[!ht]
\begin{center}
\begin{tabular}{ccc} \hline\noalign{\smallskip}
Type&Involved set& Binary answer {\smallskip} \\
\hline\hline\noalign{\smallskip}
$T_a=11111$ & $I(T_a)=\{0,1 ,2,3,4\}$ & $b(T_a)=1$\\
\hline\noalign{\smallskip}
$T_0=10100$ & $I(T_0)=\{0,1 ,4\}$ & $b(T_0)=0$\\
\hline\noalign{\smallskip}
$T_1=01010$ & $I(T_1)=\{0,1 ,2\}$ & $b(T_1)=0$ \\
\hline\noalign{\smallskip}
 $T_2=00101$ & $I(T_2)=\{1 ,2,3\}$ & $b(T_2)=0$ \\
\hline\noalign{\smallskip}
 $T_3=10010$ & $I(T_3)=\{2,3 ,4\}$ & $b(T_3)=0$ \\
 \hline\noalign{\smallskip}
 $T_4=01001$ & $I(T_4)=\{3,4 ,0\}$ & $b(T_4)=0$ \\
\hline
\end{tabular}
\end{center}
\caption{$NC_{01}(G)$ game (Here the players are identified with the integers modulo 5).}
\label{T:C01}
\end{table}

\subsection{Comments on \texorpdfstring{$NC_{00,0}(C_5)$}{Lg}}

A modification of a different kind consists in adding more 
questions from the stabiliser. As the first example of this kind we
define a game $NC_{00,0}(C_5)$, where the additional family of questions has 
four involved players with the non-involved player getting type
$0$,  as specified by Table \ref{T:C000}.

\begin{table}[!ht]
\begin{center}
\begin{tabular}{ccc} \hline\noalign{\smallskip}
Type&Involved set& Binary answer {\smallskip} \\
\hline\hline\noalign{\smallskip}
$T_a=11111$ & $I(T_a)=\{0,1 ,2,3,4\}$ & $b(T_0)=1$\\
\hline\noalign{\smallskip}
$T_{i_1}=0_{i_1-2}0_{i_1-1}1_{i_1}0_{i_1+1}0_{i_1+2} $ & $I(T_{i_1})=\{i_1-1,i_1,i_1+1\}$ & $b(T_{i_1})=0$\\
 $i_1\in \{0,\ldots ,4\}$ &&\\
\hline\noalign{\smallskip}
  $T_{i_2}=0_{i_2-1}1_{i_2} 0_{i_2+1}1_{i_2+2}0_{i_2+3}$ & $  I(T_{i_2})=\{i_2-1,i_2,i_2+2,i_2+3\}$ &$ b(T_{i_2})=0$ \\
   $i_2\in \{0,\ldots ,4\}$ &&\\ 
\hline\noalign{\smallskip}
\hline
\end{tabular}
\end{center}
\caption{$NC_{00,0}(G)$ game.} 
\label{T:C000}
\end{table}

   

For $v_1=1$, $v_0= \frac 2 3$,  and the probability distribution  $w(T_a)=\frac 3 {13}$, $w(T_{i_1})=w(T_{i_2})= \frac 1 {13}$
we get a CSW of 0.72  versus a QSW of $0.83$ 

Note that each player gets types $0$ and $1$ with different probabilities. In fact, it is simple to show that no choice of $w_1, w_2$ and $w_3$ can make 
these probabilities equal. However, it is possible to modify the set of types so that equality becomes possible, as shown in the following.

\subsection{Comments on \texorpdfstring{$NC_{00,01,0}(C_5)$}{Lg}}


We increase the set of types using other questions from the stabiliser: 
We define a game $NC_{00,01,0}(C_5)$  for which with a suitable choice 
of probability distribution the players get  0 and 1 with the same probability. 

\begin{table}[!ht]
\begin{center}
\begin{tabular}{ccc} \hline\noalign{\smallskip}
Type&Involved set& Binary answer {\smallskip} \\
\hline\hline\noalign{\smallskip}
$T_a=11111$ & $I(T_a)=\{0,1 ,2,3,4\}$ & $b(T_0)=1$\\
\hline\noalign{\smallskip}
$T_{i_1}=0_{i_1-2}0_{i_1-1}1_{i_1}0_{i_1+1}0_{i_1+2} $ & $I(T_{i_1})=\{i_1-1,i_1,i_1+1\}$ & $b(T_{i_1})=0$\\
 $i_1\in \{0,\ldots ,4\} $&& \\
 \hline\noalign{\smallskip}
 $T_{i_2}=0_{i_2-1}1_{i_2} 0_{i_2+1}1_{i_2+2}0_{i_2+3} $& $I(T_{i_2})=\{i_2-1,i_2,i_2+1\} $&$  b(T_{i_2})=0$ \\
 $i_2\in \{0,\ldots ,4\} $&&\\
\hline\noalign{\smallskip}
  $T_{i_3}=0_{i_3-1}1_{i_3} 0_{i_3+1}1_{i_3+2}0_{i_3+3}$&$  I(T_{i_3})=\{i_2-1,i_2,i_2+2,i_2+3\}$ &$ b(T_{i_3})=0$ \\
   $i_3\in \{0,\ldots ,4\} $&&\\ 
\hline\noalign{\smallskip}
\hline
\end{tabular}
\end{center}
\caption{$NC_{00,01,0}(G)$ game.} 
\label{T:C00010}
\end{table}

We consider this game with type probability distributions given by $w(T_a)=3/13$, 
$w(t_{i_1})=1/26$,   $w(T_{i_2})=1/26$ and $w(T_{i_3})=1/13$.

The involvement probabilities satisfy $P_{inv}(1)>P_{inv}(0)=8/13$ and  the best 
classical Social wellfare with $v_0=2/3$, $v_1=1$ is $CSW= 0.72$  versus a
 QSW of $0.83$ 


Note that even though the types $T_{i_2}$ and $T_{i_3}$ are similar,
the involved sets and thus the utilities are different. However, if one wants to 
restrict to scenarios in which the utility can be  deterministically determined
from the type, one can just add an extra player with a type allowing to distinguish 
the different cases and with utility the average utility of the other players 
independently of his/her action. 

\section{Quantum vs correlation separation}\label{sec:amp}

In   \cite{belief1} it is asked as an open question whether the separation 
between classical and quantum social welfare is bounded.
We show in this section how two families of amplification techniques can 
increase the separation by  adding a penalty for wrong answers and then by
increasing the number of players. 

\subsection{Wrong answer penalty}
A possible technique is to penalize  bad answers more, using the fact that classical 
functions always produce a bad answer for some question.
Instead of getting 0 when losing we generalize so that   each  player gets
$-N_g v_1$ if they answer 1 and $-N_g v_0$ if they answer 0, where $N_g$ can 
be seen as the penalty for giving a wrong answer.
If $\delta_{(s,t),{\cal L}}=1$ if $(s,t)\in {\cal L}$ and 0 otherwise, and
$N_g$ is a positive number, then 
 $$u_j(s|t)= (-N_g)^{\delta_{(s,t),{\cal L}}}    v_{s_j}$$ 

For  $NC_{01}(C_5)$ as soon as $N_g>3 v_1$ there exists only two classical Nash
equilibria:

\begin{itemize}
\item  All 0 with a  social welfare of
 $\frac{ - N_g  v_0 +
5 v_0}{6}$ and 
\item All NOT  with a social welfare of $\frac{-N_g  v_0 + 2 v_0 + 3 v_1}{6}$.
\end{itemize}


Therefore  the classical social welfare decreases linearly with the penalty while the quantum average social welfare remains $\frac{ v_1+ v_0} {2}$.


\subsection{Distributed parallel repetition}

The distributed parallel  composition of nonlocal games  appears in \cite{holmg} 
for the study of non signaling correlations and also   in \cite{Viddick}
where it is called $k$-fold repetition. 
$k$ groups of players play at the same time and they win collectively if 
all the groups win their game.

\begin{theorem} There exists games with bounded personal utilities
$v_0$, $v_1$ on $O(log(\frac 1 \epsilon))$ players
 ensuring  $\frac{CSW(G)}{QSW(G}\le \epsilon$ for the ratio   best 
 classical social welfare over quantum social welfare with guaranteed value.
\end{theorem}

\begin{proof}
It is easy to bound the utility in these settings as  for any strategy in a
repeated game. If a player $p$ is involved in the strategy $S_j$ 
but is not involved in the strategy $S_i$ of another group then his utility is 
conditioned by the fact that the $S_i$ strategy wins to receive a positive utility
and 

$$u^p(S_i\times S_j)\le p_{win}(S_i) u^p(S_j)$$

As the quantum strategy obtained from following the nonlocal advice always
wins, the QSW remains unchanged whereas the CSW decreases. 
For instance  $CSW(k-$fold$\  NC_{00}(C_5)) =\frac {5} {6}^k CSW(NC_{00}(C_5))$.

Therefore using these games one can build  games with bounded personal 
utilities $v_0$, $v_1$ on $O(log(\frac 1 \epsilon))$ players
 ensuring $\frac{CSW(G)}{QSW(G}\le \epsilon$
\end{proof}

 \section{Conclusion}
 We have used properties of multipartite graph games to define conflict of interest games, and shown
 that by combining such games  the ratio classical social welfare / quantum social welfare can go to zero.

  One can easily extend to stabilizer games \cite{Viddick} to have any number
  of types and possible strategies.
  
  As pointed out by \cite{belief1}, quantum advice equilibria can be reached
without needing a trusted mediator, furthermore they 
ensure privacy as they are belief invariant.
Some other features may be emphasized  if we define Nash equilibria using
pseudotelepathy games:  such situations ensure a guaranteed utility  and
they are also better when analysing the maximal minimal utility.
  It may be interesting to investigate further how this guaranteed value property
  for some quantum equilibria can be used.
On the other hand it would also be interesting to investigate how relaxing the guaranteed win requirement might allow to increase the QSW even further.

  The possibility of potentially unlimited
  improvement of social welfare while preserving belief invariance is therefore a 
  strong motivation to consider  classical payoff tables that arise for usual
  situations in which Nash equilibria occur and play an important role. For  
   example, in routing problems an  advice provider could  use
  a quantum advice system as follows. Send either    one rotated quantum bit to
  each player asking for the advice and each player measures his qubit to get the 
  answer, or  (in a trusted setting) compute the advice to give using a quantum measurement and send a classical message.


\section*{Acknowledgement}
This research was supported through the program ``Research in Pairs'' 
by the
Mathematisches Forschungsinstitut Oberwolfach in 2019.
The authors also acknowledge the
\emph{``Investissements d'avenir''} (ANR-15-IDEX-02) program of
the French National Research Agency, NCN grant Sonata UMO-2014/14/E/ST2/00020 and thank Sidney Sussex College, Cambridge for support.

  \newpage
\section*{Appendix}
\label{Appendix}
  
  {\bf  Pareto} equilibria  for $NC_{00}$ when $v_0/v_1\le 1/3$, Nb solutions : 121,  Nb distinct  equilibria : 18

$  
  \begin{array}{lllll|lllll|l} \multicolumn{5}{c}{\text{Local functions}} &  \multicolumn{5}{c}{\text{Players utility $[\times 6]$}}& SW [\times30] \\ \hline \\
2
&
1
&
1
&
0
&
0
&
3
v_0
\, + \, 
2
v_1
&
5
v_1
&
5
v_1
&
5
v_0
&
5
v_0
&
13v_0 +12v_1
\\
3
&
3
&
2
&
0
&
0
&
v_0
\, + \, 
2
v_1
&
v_0
\, + \, 
2
v_1
&
v_0
\, + \, 
2
v_1
&
3
v_0
&
3
v_0
&
9v_0 +6v_1
\\
3
&
2
&
1
&
1
&
0
&
2
v_0
\, + \, 
3
v_1
&
4
v_0
\, + \, 
v_1
&
5
v_1
&
5
v_1
&
5
v_0
&
11v_0 +14v_1
\\
1
&
3
&
1
&
1
&
0
&
5
v_1
&
2
v_0
\, + \, 
3
v_1
&
5
v_1
&
5
v_1
&
5
v_0
&
7v_0 +18v_1
\\
1
&
3
&
3
&
1
&
0
&
5
v_1
&
v_0
\, + \, 
4
v_1
&
v_0
\, + \, 
4
v_1
&
5
v_1
&
5
v_0
&
7v_0 +18v_1
\\
2
&
3
&
1
&
2
&
0
&
3
v_0
\, + \, 
2
v_1
&
v_0
\, + \, 
4
v_1
&
5
v_1
&
3
v_0
\, + \, 
2
v_1
&
5
v_0
&
12v_0 +13v_1
\\
3
&
2
&
2
&
3
&
0
&
v_0
\, + \, 
4
v_1
&
4
v_0
\, + \, 
v_1
&
4
v_0
\, + \, 
v_1
&
v_0
\, + \, 
4
v_1
&
5
v_0
&
15v_0 +10v_1
\\
2
&
1
&
1
&
1
&
1
&
2
v_0
\, + \, 
v_1
&
3
v_1
&
3
v_1
&
3
v_1
&
3
v_1
&
2v_0 +13v_1
\\
2
&
2
&
1
&
1
&
1
&
3
v_0
\, + \, 
2
v_1
&
3
v_0
\, + \, 
2
v_1
&
5
v_1
&
5
v_1
&
5
v_1
&
6v_0 +19v_1
\\
3
&
3
&
1
&
1
&
1
&
2
v_0
\, + \, 
v_1
&
2
v_0
\, + \, 
v_1
&
3
v_1
&
3
v_1
&
3
v_1
&
4v_0 +11v_1
\\
3
&
1
&
3
&
1
&
1
&
2
v_0
\, + \, 
v_1
&
3
v_1
&
2
v_0
\, + \, 
v_1
&
3
v_1
&
3
v_1
&
4v_0 +11v_1
\\
3
&
3
&
1
&
2
&
1
&
2
v_0
\, + \, 
3
v_1
&
2
v_0
\, + \, 
3
v_1
&
5
v_1
&
4
v_0
\, + \, 
v_1
&
5
v_1
&
8v_0 +17v_1
\\
3
&
3
&
2
&
2
&
1
&
2
v_0
\, + \, 
3
v_1
&
v_0
\, + \, 
4
v_1
&
3
v_0
\, + \, 
2
v_1
&
3
v_0
\, + \, 
2
v_1
&
5
v_1
&
9v_0 +16v_1
\\
3
&
2
&
3
&
2
&
1
&
v_0
\, + \, 
2
v_1
&
2
v_0
\, + \, 
v_1
&
2
v_0
\, + \, 
v_1
&
2
v_0
\, + \, 
v_1
&
3
v_1
&
7v_0 +8v_1
\\
3
&
2
&
2
&
3
&
1
&
v_0
\, + \, 
2
v_1
&
v_0
\, + \, 
2
v_1
&
v_0
\, + \, 
2
v_1
&
v_0
\, + \, 
2
v_1
&
3
v_1
&
4v_0 +11v_1
\\
3
&
3
&
3
&
3
&
1
&
2
v_0
\, + \, 
3
v_1
&
2
v_0
\, + \, 
3
v_1
&
2
v_0
\, + \, 
3
v_1
&
2
v_0
\, + \, 
3
v_1
&
5
v_1
&
8v_0 +17v_1
\\
3
&
2
&
3
&
2
&
2
&
2
v_0
\, + \, 
3
v_1
&
4
v_0
\, + \, 
v_1
&
2
v_0
\, + \, 
3
v_1
&
3
v_0
\, + \, 
2
v_1
&
3
v_0
\, + \, 
2
v_1
&
14v_0 +11v_1
\\
3
&
3
&
3
&
3
&
3
&
v_0
\, + \, 
4
v_1
&
v_0
\, + \, 
4
v_1
&
v_0
\, + \, 
4
v_1
&
v_0
\, + \, 
4
v_1
&
v_0
\, + \, 
4
v_1
&
5v_0 +20v_1
\\
\end{array}
  $
  
\vspace{0.4cm}
\noindent {\bf Pareto} equilibria for $NC_{00}$ when $1/3\le v_0/v_1\le 1/2$ Nb solutions : 91, Nb distinct  equilibria : 14

$\begin{array}{lllll|lllll|l} \multicolumn{5}{c}{\text{Local functions}} &  \multicolumn{5}{c}{\text{Players utility $[\times 6]$}}& SW [\times30] \\ \hline \\
2
&
1
&
1
&
0
&
0
&
3
v_0
\, + \, 
2
v_1
&
5
v_1
&
5
v_1
&
5
v_0
&
5
v_0
&
13v_0 +12v_1
\\
3
&
2
&
1
&
1
&
0
&
2
v_0
\, + \, 
3
v_1
&
4
v_0
\, + \, 
v_1
&
5
v_1
&
5
v_1
&
5
v_0
&
11v_0 +14v_1
\\
1
&
3
&
1
&
1
&
0
&
5
v_1
&
2
v_0
\, + \, 
3
v_1
&
5
v_1
&
5
v_1
&
5
v_0
&
7v_0 +18v_1
\\
1
&
3
&
3
&
1
&
0
&
5
v_1
&
v_0
\, + \, 
4
v_1
&
v_0
\, + \, 
4
v_1
&
5
v_1
&
5
v_0
&
7v_0 +18v_1
\\
2
&
3
&
1
&
2
&
0
&
3
v_0
\, + \, 
2
v_1
&
v_0
\, + \, 
4
v_1
&
5
v_1
&
3
v_0
\, + \, 
2
v_1
&
5
v_0
&
12v_0 +13v_1
\\
3
&
2
&
2
&
3
&
0
&
v_0
\, + \, 
4
v_1
&
4
v_0
\, + \, 
v_1
&
4
v_0
\, + \, 
v_1
&
v_0
\, + \, 
4
v_1
&
5
v_0
&
15v_0 +10v_1
\\
2
&
2
&
1
&
1
&
1
&
3
v_0
\, + \, 
2
v_1
&
3
v_0
\, + \, 
2
v_1
&
5
v_1
&
5
v_1
&
5
v_1
&
6v_0 +19v_1
\\
3
&
3
&
1
&
1
&
1
&
2
v_0
\, + \, 
v_1
&
2
v_0
\, + \, 
v_1
&
3
v_1
&
3
v_1
&
3
v_1
&
4v_0 +11v_1
\\
3
&
3
&
1
&
2
&
1
&
2
v_0
\, + \, 
3
v_1
&
2
v_0
\, + \, 
3
v_1
&
5
v_1
&
4
v_0
\, + \, 
v_1
&
5
v_1
&
8v_0 +17v_1
\\
3
&
3
&
2
&
2
&
1
&
2
v_0
\, + \, 
3
v_1
&
v_0
\, + \, 
4
v_1
&
3
v_0
\, + \, 
2
v_1
&
3
v_0
\, + \, 
2
v_1
&
5
v_1
&
9v_0 +16v_1
\\
3
&
2
&
2
&
3
&
1
&
v_0
\, + \, 
2
v_1
&
v_0
\, + \, 
2
v_1
&
v_0
\, + \, 
2
v_1
&
v_0
\, + \, 
2
v_1
&
3
v_1
&
4v_0 +11v_1
\\
3
&
3
&
3
&
3
&
1
&
2
v_0
\, + \, 
3
v_1
&
2
v_0
\, + \, 
3
v_1
&
2
v_0
\, + \, 
3
v_1
&
2
v_0
\, + \, 
3
v_1
&
5
v_1
&
8v_0 +17v_1
\\
3
&
2
&
3
&
2
&
2
&
2
v_0
\, + \, 
3
v_1
&
4
v_0
\, + \, 
v_1
&
2
v_0
\, + \, 
3
v_1
&
3
v_0
\, + \, 
2
v_1
&
3
v_0
\, + \, 
2
v_1
&
14v_0 +11v_1
\\
3
&
3
&
3
&
3
&
3
&
v_0
\, + \, 
4
v_1
&
v_0
\, + \, 
4
v_1
&
v_0
\, + \, 
4
v_1
&
v_0
\, + \, 
4
v_1
&
v_0
\, + \, 
4
v_1
&
5v_0 +20v_1
\\
\end{array}$

\vspace{0.4cm}
\noindent {\bf Pareto} equilibria when $\ge 1/2$, NB solutions 81,  Nb distinct  equilibria : 12

$\begin{array}{lllll|lllll|l} \multicolumn{5}{c}{\text{Local functions}} &  \multicolumn{5}{c}{\text{Players utility $[\times 6]$}}& SW [\times30] \\ \hline \\
2
&
1
&
1
&
0
&
0
&
3
v_0
\, + \, 
2
v_1
&
5
v_1
&
5
v_1
&
5
v_0
&
5
v_0
&
13v_0 +12v_1
\\
3
&
2
&
1
&
1
&
0
&
2
v_0
\, + \, 
3
v_1
&
4
v_0
\, + \, 
v_1
&
5
v_1
&
5
v_1
&
5
v_0
&
11v_0 +14v_1
\\
1
&
3
&
1
&
1
&
0
&
5
v_1
&
2
v_0
\, + \, 
3
v_1
&
5
v_1
&
5
v_1
&
5
v_0
&
7v_0 +18v_1
\\
1
&
3
&
3
&
1
&
0
&
5
v_1
&
v_0
\, + \, 
4
v_1
&
v_0
\, + \, 
4
v_1
&
5
v_1
&
5
v_0
&
7v_0 +18v_1
\\
2
&
3
&
1
&
2
&
0
&
3
v_0
\, + \, 
2
v_1
&
v_0
\, + \, 
4
v_1
&
5
v_1
&
3
v_0
\, + \, 
2
v_1
&
5
v_0
&
12v_0 +13v_1
\\
3
&
2
&
2
&
3
&
0
&
v_0
\, + \, 
4
v_1
&
4
v_0
\, + \, 
v_1
&
4
v_0
\, + \, 
v_1
&
v_0
\, + \, 
4
v_1
&
5
v_0
&
15v_0 +10v_1
\\
2
&
2
&
1
&
1
&
1
&
3
v_0
\, + \, 
2
v_1
&
3
v_0
\, + \, 
2
v_1
&
5
v_1
&
5
v_1
&
5
v_1
&
6v_0 +19v_1
\\
3
&
3
&
1
&
2
&
1
&
2
v_0
\, + \, 
3
v_1
&
2
v_0
\, + \, 
3
v_1
&
5
v_1
&
4
v_0
\, + \, 
v_1
&
5
v_1
&
8v_0 +17v_1
\\
3
&
3
&
2
&
2
&
1
&
2
v_0
\, + \, 
3
v_1
&
v_0
\, + \, 
4
v_1
&
3
v_0
\, + \, 
2
v_1
&
3
v_0
\, + \, 
2
v_1
&
5
v_1
&
9v_0 +16v_1
\\
3
&
3
&
3
&
3
&
1
&
2
v_0
\, + \, 
3
v_1
&
2
v_0
\, + \, 
3
v_1
&
2
v_0
\, + \, 
3
v_1
&
2
v_0
\, + \, 
3
v_1
&
5
v_1
&
8v_0 +17v_1
\\
3
&
2
&
3
&
2
&
2
&
2
v_0
\, + \, 
3
v_1
&
4
v_0
\, + \, 
v_1
&
2
v_0
\, + \, 
3
v_1
&
3
v_0
\, + \, 
2
v_1
&
3
v_0
\, + \, 
2
v_1
&
14v_0 +11v_1
\\
3
&
3
&
3
&
3
&
3
&
v_0
\, + \, 
4
v_1
&
v_0
\, + \, 
4
v_1
&
v_0
\, + \, 
4
v_1
&
v_0
\, + \, 
4
v_1
&
v_0
\, + \, 
4
v_1
&
5v_0 +20v_1
\\
\end{array}$
\newpage

{\bf Nash equilibria for  $NC_{01}$}

Nash equilibria for $NC_{01}$ when $1/3\le v_0/v_1\le  1/2$, Nb solutions : 76, Nb distinct  equilibria : 13

$\begin{array}{lllll|lllll|l} \multicolumn{5}{c}{\text{Local functions}} &  \multicolumn{5}{c}{\text{Players utility $[\times 6]$}}& SW [\times30] \\ \hline \\
1
&
1
&
2
&
0
&
0
&
5
v_1
&
5
v_1
&
2
v_0
\, + \, 
3
v_1
&
5
v_0
&
5
v_0
&
12v_0 +13v_1
\\
3
&
2
&
1
&
1
&
0
&
3
v_0
\, + \, 
2
v_1
&
3
v_0
\, + \, 
2
v_1
&
5
v_1
&
5
v_1
&
5
v_0
&
11v_0 +14v_1
\\
1
&
3
&
1
&
1
&
0
&
5
v_1
&
3
v_0
\, + \, 
2
v_1
&
5
v_1
&
5
v_1
&
5
v_0
&
8v_0 +17v_1
\\
3
&
2
&
2
&
1
&
0
&
2
v_0
\, + \, 
3
v_1
&
2
v_0
\, + \, 
3
v_1
&
2
v_0
\, + \, 
3
v_1
&
5
v_1
&
5
v_0
&
11v_0 +14v_1
\\
1
&
3
&
3
&
1
&
0
&
5
v_1
&
2
v_0
\, + \, 
3
v_1
&
2
v_0
\, + \, 
3
v_1
&
5
v_1
&
5
v_0
&
9v_0 +16v_1
\\
1
&
3
&
1
&
2
&
0
&
5
v_1
&
2
v_0
\, + \, 
3
v_1
&
5
v_1
&
2
v_0
\, + \, 
3
v_1
&
5
v_0
&
9v_0 +16v_1
\\
2
&
1
&
3
&
2
&
0
&
2
v_0
\, + \, 
3
v_1
&
5
v_1
&
2
v_0
\, + \, 
3
v_1
&
2
v_0
\, + \, 
3
v_1
&
5
v_0
&
11v_0 +14v_1
\\
2
&
2
&
1
&
1
&
1
&
3
v_0
\, + \, 
2
v_1
&
2
v_0
\, + \, 
3
v_1
&
5
v_1
&
5
v_1
&
5
v_1
&
5v_0 +20v_1
\\
3
&
3
&
1
&
2
&
1
&
2
v_0
\, + \, 
3
v_1
&
3
v_0
\, + \, 
2
v_1
&
5
v_1
&
3
v_0
\, + \, 
2
v_1
&
5
v_1
&
8v_0 +17v_1
\\
3
&
3
&
2
&
2
&
1
&
3
v_0
\, + \, 
2
v_1
&
2
v_0
\, + \, 
3
v_1
&
2
v_0
\, + \, 
3
v_1
&
3
v_0
\, + \, 
2
v_1
&
5
v_1
&
10v_0 +15v_1
\\
3
&
3
&
3
&
3
&
1
&
3
v_0
\, + \, 
2
v_1
&
2
v_0
\, + \, 
3
v_1
&
3
v_0
\, + \, 
2
v_1
&
3
v_0
\, + \, 
2
v_1
&
5
v_1
&
11v_0 +14v_1
\\
3
&
2
&
3
&
2
&
2
&
3
v_0
\, + \, 
2
v_1
&
3
v_0
\, + \, 
2
v_1
&
3
v_0
\, + \, 
2
v_1
&
3
v_0
\, + \, 
2
v_1
&
2
v_0
\, + \, 
3
v_1
&
14v_0 +11v_1
\\
3
&
3
&
3
&
3
&
3
&
2
v_0
\, + \, 
3
v_1
&
2
v_0
\, + \, 
3
v_1
&
2
v_0
\, + \, 
3
v_1
&
2
v_0
\, + \, 
3
v_1
&
2
v_0
\, + \, 
3
v_1
&
10v_0 +15v_1
\\
\end{array}$

\vspace{0.4cm}
Nash equilibria for $NC_{01}$ when $v_0/v_1\ge 1/2$,  Nb solutions 40,  Nb distinct  equilibirums :6

$\begin{array}{lllll|lllll|l} \multicolumn{5}{c}{\text{Local functions}} &  \multicolumn{5}{c}{\text{Players utility $[\times 6]$}}& SW [\times30]  \\ \hline \\
3
&
2
&
1
&
1
&
0
&
3
v_0
\, + \, 
2
v_1
&
3
v_0
\, + \, 
2
v_1
&
5
v_1
&
5
v_1
&
5
v_0
&
11v_0 +14v_1
\\
1
&
3
&
1
&
1
&
0
&
5
v_1
&
3
v_0
\, + \, 
2
v_1
&
5
v_1
&
5
v_1
&
5
v_0
&
8v_0 +17v_1
\\
2
&
2
&
1
&
1
&
1
&
3
v_0
\, + \, 
2
v_1
&
2
v_0
\, + \, 
3
v_1
&
5
v_1
&
5
v_1
&
5
v_1
&
5v_0 +20v_1
\\
3
&
3
&
1
&
2
&
1
&
2
v_0
\, + \, 
3
v_1
&
3
v_0
\, + \, 
2
v_1
&
5
v_1
&
3
v_0
\, + \, 
2
v_1
&
5
v_1
&
8v_0 +17v_1
\\
3
&
3
&
3
&
3
&
1
&
3
v_0
\, + \, 
2
v_1
&
2
v_0
\, + \, 
3
v_1
&
3
v_0
\, + \, 
2
v_1
&
3
v_0
\, + \, 
2
v_1
&
5
v_1
&
11v_0 +14v_1
\\

3
&
2
&
3
&

2
&

2
&
3
v_0
\, + \, 
2
v_1
&
3
v_0
\, + \, 
2
v_1
&
3
v_0
\, + \, 
2
v_1
&
3
v_0
\, + \, 
2
v_1
&
2
v_0
\, + \, 
3
v_1
&
14v_0 +11v_1
\\
\end{array}$

\bibliographystyle{plain}
\bibliography{oquantumsfsimplified}

\end{document}